\documentclass[a4paper,UKenglish]{llncs}

\usepackage{amsmath}
\usepackage{xspace}
\usepackage{amsfonts}
\usepackage{amssymb}
\usepackage{proof}
\usepackage{empheq}
\usepackage{graphicx}
\usepackage{bibentry}
\usepackage{hyperref}
\usepackage{tikz-cd}
\usepackage{tabulary}
\usepackage{wrapfig}
\usepackage{bbm}
\usepackage{multirow}
\usetikzlibrary{automata, arrows,shadows,shapes,matrix,decorations,calc,positioning,shapes,fit}
\usepackage{algorithmic}
\usepackage{algorithm}

\makeatletter
\newcommand\footnoteref[1]{\protected@xdef\@thefnmark{\ref{#1}}\@footnotemark}
\makeatother

\newcommand{\newclass}[2]{\newcommand{#1}{\textsc{#2}}}
\newclass{\exptime}{ExpTime}
\newclass{\etime}{ETime}
\newclass{\nexptime}{NExpTime}
\newclass{\class}{Class}
\newclass{\pspace}{PSPACE}
\newclass{\expspace}{EXPSPACE}
\newclass{\dtime}{DTime}
\newclass{\dspace}{DSpace}

\newcommand{\TS}{\mathcal{TS}}
\newcommand{\Confs}{\Gamma}
\newcommand{\Labels}{\mathcal{L}}
\newcommand{\Trans}{\mathcal{T}}

\newcommand{\context}[1]{\textsf{#1}}

\newcommand{\pmc}{\context{CB}}
\newcommand{\ch}{\mathsf{ch}}

\newcommand{\Relord}{\Rel_{\leq n}}
\newcommand{\Relinit}{\Rel_{<}}

\newcommand{\rtmp}{r_{tmp}}
\newcommand{\rnew}{r_{\$}}
\newcommand{\rchtmp}{r_{tmp}^{\ch}}
\newcommand{\rchnew}{r_{\$}^{\ch}}

\newcommand{\funQ}{\mathsf{St}}
\newcommand{\funX}{\mathsf{Mem}}
\newcommand{\funB}{\mathsf{Buf}}
\newcommand{\funR}{\mathsf{RVal}}
\newcommand{\funV}{\mathsf{XVal}}

\newcommand{\prog}{\mathsf{Prog}}
\newcommand{\conf}{\gamma}
\newcommand{\initconf}{\gamma_{\mathsf{init}}}

\newcommand{\init}{\mathsf{init}}

\newcommand{\funXinit}{\funX_{\init}}
\newcommand{\funBinit}{\funB_{\init}}
\newcommand{\funRinit}{\funR_{\init}}

\newcommand{\trans}[2]{\xrightarrow{#1 ,#2}}
\newcommand{\tstrans}[1]{\xrightarrow{#1}}

\newcommand{\D}{\mathtt{D}}
\newcommand{\X}{\mathcal{X}}
\newcommand{\Rel}{\mathsf{Rl}}
\newcommand{\Q}{\mathcal{Q}}
\newcommand{\R}{\mathcal{R}}
\newcommand{\T}{\mathcal{T}}
\newcommand{\N}{\mathbb{N}}
\newcommand{\B}{\mathbb{B}}

\newcommand{\scbabname}{\mathtt{AB}}
\newcommand{\scbab}{\scbabname(\prog,k )}

\newcommand{\relAB}{\Relinit{-}\scbab}

\newcommand{\reachword}{\texttt{Reach}}
\newcommand{\greach}{\reachword\xspace}

\newcommand{\reachof}[2]{$#1$-\reachword$[ #2 ]$}
\newcommand{\greachof}[1]{\reachword$[ #1 ]$}
\newcommand\run{\mathbb{\pi}}

\newcommand{\final}{\mathit{final}}
\newcommand{\qfinal}{q_{\final}}
\newcommand\define{\mathrel{:=}}
\newcommand{\Op}{\mathsf{Op}}
\newcommand{\op}{\mathsf{op}}
\newcommand{\wop}[2]{\mathsf{wt}( #1 , #2 )}
\newcommand{\rop}[2]{\mathsf{rd}( #1 , #2 )}
\newcommand{\rl}[2]{\mathsf{rl}( #1 , #2 )}
\newcommand{\arw}[3]{\mathsf{arw}(#1, #2, #3)}
\newcommand{\rlvar}{\mathsf{rl}}
\newcommand\newval{\define \circledast}

\newcommand{\lcs}{\mathcal{L}}

\newcommand{\fsm}{\mathcal{A}}

\newcommand\xes{\textsc{x86}}
\newcommand\tso{{TSO}}

\newcommand{\relname}[1]{\textsf{#1}}

\newcommand\arel[3]{\langle #1, #2, #3 \rangle}

\title{Verification under \tso\ with an infinite Data Domain}

\institute{Uppsala University, Sweden 
	\and Technical University of Denmark, Denmark
	\and Indian Institute of Technology Bombay, India
}
\author{Parosh Aziz Abdulla\inst{1} \and
	Mohamed Faouzi Atig	\inst{1} \and
	Florian Furbach\inst{2}  \and
	Shashwat Garg  \inst{3}}
\authorrunning{Parosh A. Abdulla et al.}

\begin{document}
	\maketitle
	\begin{abstract}
	We examine verification of concurrent programs under the total store ordering (\tso) semantics used by the \xes\ architecture.
	In our model, threads manipulate variables over infinite domains and they can check whether variables are related for a range of relations.
	We show that, in general, the control state reachability problem is undecidable. This result is derived through a reduction from the state reachability problem of lossy channel systems with data  (which is  known to be undecidable). 

	In the light of this undecidability, we turn our attention to a more tractable variant of the reachability problem. Specifically, we study context bounded runs, which provide an under-approximation of the program behavior by limiting the possible interactions between processes. A run consists of a number of contexts, 
	with each context representing a sequence of steps where a only single designated thread is active. 
	We prove that the control state reachability problem under bounded context switching  is \pspace\ complete.
\end{abstract}

	\section{Introduction}

Over the years, research on concurrent verification has been chiefly conducted under the premise that the threads run according to the classical Sequential Consistency (SC) semantics.
Under SC, the threads operate on a set of shared variables through which they communicate \emph{atomically}, i.e., read and write operations take effect immediately.
In particular, a write operation is visible to all the threads as soon as the writer thread carries out its operation.
Therefore, the threads always maintain a uniform view of the shared memory: they all see the latest value written on any given variable and we can interpret program runs as interleavings of sequential thread executions.
Although SC has been immensely popular as an intuitive way of understanding the behaviours of concurrent threads, it is not realistic to assume computation platforms guarantee SC anymore.
The reason is that, due to hardware and compiler optimizations, most modern platforms allow more relaxed program behaviours than those permitted under SC,
leading to so-called \emph{weak memory models}.
Weakly consistent platforms are found at all levels of system design
such as multiprocessor architectures
(e.g., \cite{SSONM2010,SarkarSAMW11}),
Cache protocols (e.g., \cite{aros-micro16,DBLP:conf/hpca/ElverN14}),
language level concurrency (e.g., \cite{LahavGV16}), and
distributed data stores (e.g., \cite{Burckhardt14}).
Program behaviours change dramatically when moving from the SC semantics to weaker semantics.
Therefore, in recent years, research on the verification of concurrent programs under weak memory models
have started to become popular.
A classical example of weak memory models is the Total Store Ordering (\tso) semantics which is a formalization of the Intel x86 processor architecture \cite{OSS2009}.
The \tso\ semantics inserts an unbounded FIFO buffer, called the {\it store buffer}, between each thread and the main memory.
When a thread performs a write instruction, the corresponding operation is appended to end of the buffer, and hence it is not immediately visible to other threads.
The write messages are non-deterministically propagated from the store buffer of a given thread to the shared memory.
Verification of programs that contain data races needs to take the underlying memory model into account.
This is crucial in hardware-close programming, especially in concurrent libraries or kernels. 
Such applications are inherently racy; exploiting racy WMM operations for efficiency is standard practice.
Our work serves as a foundation for ensuring the correctness of such systems, which often rely on these intricate memory models to achieve optimal performance.

In a parallel development, significant research has been done on extending model checking frameworks to programs with infinite state spaces.
There are two main reasons why a program might have an infinite state space.
The first is that the program has unbounded control structures, which means it can have an unbounded number of threads.
Examples include parameterized systems, in which correctness of the system is checked regardless of the number of threads, and programs that allow dynamic thread creation through spawning
\cite{DBLP:reference/mc/AbdullaST18}.
Secondly, the program may operate on unbounded data structures, such as clocks \cite{AlurD94}, stacks \cite{BEM97}, and queues (\cite{DBLP:conf/lics/AbdullaJ93,10.1145/2933575.2934535}).
These works, including their extensions, have been done under the SC assumption.
Although recent works have started to explore parameterized verification for weak memory models \cite{DBLP:journals/pacmpl/AbdullaAR20,DBLP:conf/tacas/AbdullaAFGHKS23,DBLP:conf/podc/KrishnaG0C22}, the verification of programs that operate on a shared unbounded data structure with weak memory semantics has remained unexplored until now.

In this paper, we combine infinite-state programs with weak memory models: we study the decidability and complexity of the reachability problem for programs operating on unbounded data structures under the TSO semantics.
While the \tso\ semantics has been extensively studied (e.g., \cite{DBLP:conf/esop/BouajjaniDM13,DBLP:conf/esop/AbdullaAP15}), it has been assumed that the data domain is finite. 
This means that the possible values of a shared variable or a register are bounded. 
In contrast, our model allows for an infinite domain such as natural numbers $\N$ or real numbers $\mathbb{R}$. 
It contains register assignments, an operator that may assign an arbitrary value to a register, and a set of relations that act as guards.
We focus on relations equality and "greater than" on totally ordered sets and combinations, negations and inversions of them.
Our model finds practical utility in continuously running concurrent protocols. 
A prime example is the bakery ticket protocol used in various scenarios. It is presented in \autoref{sec:app:lamport}.
Here, an unbounded number of requests occur, each assigned increasing numbers and the lowest-numbered request is serviced. 
This presents a scenario with inherent races that requires an infinite domain which our model can effectively verify. 
Note that our model is infinite in multiple dimensions: the threads are infinite-state as they operate on unbounded data domains, the store buffers are unbounded, and they carry write-messages over an unbounded domain.

In order to perform safety verification, we need to decide whether there is an execution that can reach some undesirable control state.
We study the control state reachability problem and show that for many domains and relations, it is undecidable.
Therefore, we propose an alternative approach by introducing an under-approximation schema using context-bounding \cite{DBLP:conf/tacas/QadeerR05,MQ07,DBLP:journals/fmsd/LalR09,DBLP:conf/cav/TorreMP09,ABP2011}. 
Context-bounding has been proposed in \cite{DBLP:conf/tacas/QadeerR05} as a suitable approach for efficient bug detection in multithreaded programs. Indeed, for concurrent programs, a bounding concept that provides both good coverage and scalability must be based on aspects related to the interactions between concurrent components. It has been shown experimentally that concurrency bugs usually show up  after a small number of context switches \cite{MQ07}.
In this work, we study a context bounded analysis where only the active thread may perform an operation and update the memory.
We show that in this case, the state reachability problem is not only decidable, but even \pspace\ complete. 
To this end, we perform a two-step abstraction that employs insights about context bounded runs of \tso\ semantics as well as the structure of reachable configurations.

In the first step of our abstraction process, we refine the methods introduced by \cite{ABP2011}. Their construction introduces a code-to-code translation that abstracts the buffer, simplifying the problem to state reachability under SC.
%
%
%
Our approach leverages the fact that this abstraction does not explicitly depend on variable values. 
In our case, the abstraction step yields a register machine where the register values are integers or real numbers, and the transitions are conditioned by ``gap-constraints" \cite{avis06,DBLP:conf/icalp/Cerans94,DBLP:journals/fuin/LazicNORW08}. 
Gap constraints serve to identify, within each system configuration, (i) the variables with identical values and (ii) the gaps (differences) between variable values. 
Notably, these gaps can be arbitrarily large.
The papers \cite{avis06,DBLP:conf/icalp/Cerans94,DBLP:journals/fuin/LazicNORW08} analyze programs with gap constraints within the framework of well-structured systems \cite{DBLP:journals/iandc/AbdullaCJT00,DBLP:journals/tcs/FinkelS01}. As a result, they do not provide upper bounds on the complexity.

As another key contribution of this paper, we propose a method to achieve \pspace\ completeness. 
The fundamental idea behind our algorithm is that for any system execution, there is an alternative execution with larger gaps among the variables. 
This implies that we do not need to explicitly track the gaps between variables, as is the case in \cite{avis06,DBLP:conf/icalp/Cerans94,DBLP:journals/fuin/LazicNORW08}. 
Instead, we implement a second (precise) abstraction step, focusing solely on the order of variables. For any pair of variables $x$ and $y$, we record whether $x=y$, $x<y$, or $x>y$.

	\section{Related Work}
Not much current work considers the complexity and decidability of infinite-state state programs on weak memory models.
Furthermore, most existing works consider parameterized verification rather than programs with infinite data domains.
The paper \cite{DBLP:journals/pacmpl/AbdullaAR20} considers
parameterized verification of programs running
under TSO, and shows that the reachability problem is \pspace\ complete.
However, the work assumes that the threads are finite-state and, in particular, the threads do not manipulate unbounded data domains.
The paper \cite{DBLP:conf/podc/KrishnaG0C22} shows \pspace\ completeness when the underlying semantics is the Release-Acquire fragment of C11.
The latter semantics gives rise to a different semantics compared to \tso.
The paper also considers finite-state threads.

In \cite{DBLP:journals/lmcs/AbdullaABN18},
 parameterized verification of programs running
under \tso\ is considered.
However, the paper applies the framework of well-structured systems
where the buffers of the threads are modelled as lossy channels,
and hence
the complexity of the algorithm is non-primitive recursive.
In particular, the paper does not give
any complexity bounds for the reachability
problem (or any other verification problems).
The paper \cite{DBLP:conf/esop/BouajjaniDM13} considers checking the robustness
property against SC for parameterized systems running
under the \tso\ semantics.
However, the robustness problem is entirely different
from reachability and the techniques and results developed in this work
cannot be applied in our setting.

The paper \cite{DBLP:conf/tacas/AbdullaAFGHKS23} considers parameterized verification under the TSO semantics when the individual threads are infinite-state.
However, the authors study a {\it restricted} model, where it assumes that (i) all threads are identical and (ii) the threads do not use atomic operations.
Generally, parameterized verification for the restricted model is easier than non-parameterized verification.
For instance, in the case of \tso\ where the threads are finite-state, the restricted parameterized verification problem is in \pspace\ \cite {DBLP:journals/pacmpl/AbdullaAR20} while the non-parameterized problem has a non-primitive recursive complexity \cite{DBLP:conf/popl/AtigBBM10}.

The are many works on extending infinite-state systems with unbounded data domains.
Well studied examples are Petri nets with data tokens \cite{DBLP:journals/fuin/LazicNORW08}, stacks with unbounded stack alphabets \cite{DBLP:conf/lics/AbdullaAS12}, and lossy channel systems with unbounded message alphabets \cite{10.1145/2933575.2934535}.
All these works assume the SC semantics and are hence orthogonal to this work.
	\section{Total Store Order (\tso)}\label{sec:tso}
Let $\B=\{true,false\}$.  
Given a function $f: A\rightarrow B$ with $a\in A,b\in B$,  $f[a\leftarrow b]$ is defined as follows: $f[a\leftarrow b](a)\define b$, $f[a\leftarrow b](a')\define f(a')$ for any $a'\in A$ with $a'\neq a$. 
We write $x\in w$ for letter $x\in \Sigma$ occurring in word $w\in \Sigma^*$ and $w'\leq w$ for  $w'\in \Sigma^*$ being a subsequence of $w$.

Let $x$ and $y$ be two natural (real) numbers. Let $n \in \N$, we use $x<_n y$ (resp. $ \leq_n y$) to denote that  $x+n<y$ (resp.  $x+n \leq y$). A data theory   is defined by a pair $(\D,  \Rel )$ where $\D$ is an infinite data domain and $ \Rel \subseteq \D \times \D \rightarrow \B$ is a finite set of relations over $\D$. In this paper, we restrict ourselves to the set of natural/real numbers as data domain, and the set of relations $\Rel$ to be  a subset of  $\Relord=\{ =, \neq , <, \leq, <_n, \leq_n \mid n\in \N \}$. We assume w.l.o.g. that $0 \in \D$.

\paragraph*{Transition Systems}
A labelled transition system is a tuple $\TS=(\Confs, \Labels, \Trans,\initconf)$ that consists of a set of  \emph{configurations} $\Confs$, a finite set of labels $\Labels$, a labelled transition relation $\Trans \subseteq \Confs \times \Labels \times \Confs$, and an initial configuration $\initconf\in \Confs$.
We write $\conf \tstrans{\ell} \conf' $ for $\langle \conf,\ell,\conf' \rangle \in \Trans$.
We say that $\run = t_1\ldots t_n \in \Trans^*$ is a run of $\TS$ if  there is a sequence of configurations $\conf_1, \conf_2, \ldots, \conf_{n+1}$ such that $t_i=\conf_i\tstrans{\ell_i} \conf_{i+1}$ for $i\leq n$ and $\conf_1=\initconf$.
The run $\run$ ends in configuration $\conf_{n+1}$.
We say that $\conf$ is reachable if there is a run $\run$ of $\TS$ that ends in $\conf$.

\paragraph*{Programs}
A concurrent program $\prog$ consists of finite set of threads $\T$.
Each thread $t\in \T$ is a finite state machine that works on its own set of local registers $\R_t$. The local registers of different threads are disjoint. Let $\R=\cup_{t \in \T}\R_t$. 	
The threads communicate over a finite  set of shared variables $\X$. The registers and the shared variables take their values from a data theory  $(\D,\Rel)$. 
Formally, a thread is a tuple $t=\langle \Q_t, \R_t, \Delta_t, q_{\init}^t\rangle $ where $\Q_t$ is a finite set of  states of thread $t$, $q_{\init}^t \in \Q_t$ is the initial state of $t$, and $\Delta_t\subseteq \Q_t \times \Op \times \Q_t$ is a finite set of transitions that change the  state and execute an operation $\op\in \Op$. Let $x\in \X, r_1,r_2 \in \R_t$. A transition $\delta\in \Delta_t$ is a tuple $\delta =\langle q,\op,q'\rangle $ where the operation $\op \in \Op$ has one of the following forms: $(1)$ $r_1 \define r_2$ assigns the value of register $r_2$ to register $r_1$, $(2)$ $r_1 \newval$ non-deterministically assigns a value to register $r_1$, $(3)$ $\rl{r_1}{r_2}$   checks if the  values of the two registers $r_1$ and $r_2$ satisfy the relation  $\mathsf{rl} \in \Rel$,
$(4)$ $\rop x {r_1}$ reads the value of shared variable $x$ and stores it in register $r_1$, $(5)$ $ \wop x {r_1}$ writes the value of register $r_1$ to shared variable $x$, and $(6)$ $\arw{x}{r_1}{r_2}$	 is the atomic read write operation which atomically executes a read followed by a write operation. 

\begin{figure*}[tb]
\centering
\includegraphics[width=\textwidth]{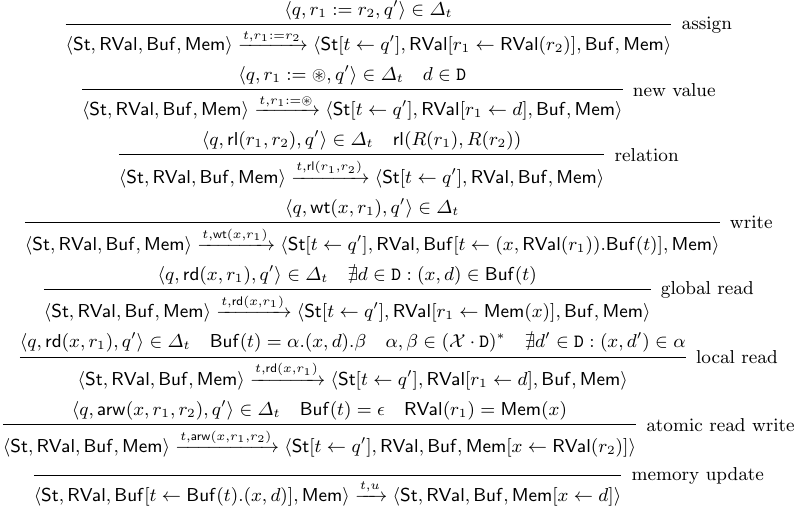}
\caption{The transition relation of  TSO. We assume that  $\funQ(t)=q$.}
\label{fig:rules_tso}
\end{figure*}

\paragraph*{\tso\ Semantics}
The \tso\ memory model \cite{SSONM2010} is used by the x86 processor architecture. 
Each thread has its own FIFO write buffer. 
Write operations $\wop{x}{r}$ in a thread $t$ do not update the memory immediately; if $d\in \D$ is the value of $r$, then $(x, d)$ is appended to the buffer of $t$. The buffer contents are updated to the shared memory non-deterministically. A read operation $\rop{x}{r}$ in $t$ accesses the latest write in the buffer of 
$t$. In case there is no such write, it accesses the shared memory.
For the atomic read write operation $\arw{x} {r_1}{ r_2}$ in thread $t$, 
 the buffer of $t$ must be empty ($\epsilon$), and the value of $x$ in the memory must be same as the value of $r_1$. Then $x$  is set to  the value of $r_2$.   

Formally, the \tso\ memory model is a labelled transition system. A configuration $\conf$ is defined as a tuple $\conf=\langle \funQ,\funR,\funB,\funX \rangle $ where  
$\funQ: \T \rightarrow \bigcup_{t\in \T} \Q_t$ maps each thread to its current state, $\funR: \R\rightarrow \D$ maps each register in a thread to its current value,  $\funB: \T \rightarrow (\X \times \D )^*$ maps each thread buffer to its content, which is a sequence of writes. 
Finally, $\funX: \X \rightarrow \D$ maps 
each shared variable to its current value in the memory. 
The initial configuration of $\prog$ is defined by a tuple 
$\initconf=\langle \funQ_{\init},\funR_{\init},\funB_{\init},\funX_{\init} \rangle$ where 
$\funQ_{\init}$ maps  each thread $t$ to its initial  states $q_{\init}^t$,  
$\funRinit$ and $ \funXinit$ assign all registers and shared variables the  value $0$, and 
$\funBinit$ initializes all thread buffers to the empty word $\epsilon$.   
We formally define the labelled  transition relation $\tstrans{\ell} $ on configurations in \autoref{fig:rules_tso} where the label $\ell$   is either of the form $t,\op$ (to denote a thread operation)    or  $t,u$ (to denote an update operation) with $ t \in \T$ is a thread and $\op \in \Op$ is an operation.

\paragraph*{The Reachability Problem \greach}
Given a concurrent program $\prog$
and a state $\qfinal \in \Q_t$ of thread $t$, 
\greach asks, if a configuration  $\conf =\langle \funQ,\funR,\funB,\funX \rangle$ with $\funQ(t)=\qfinal$ is reachable by the transition system given by the \tso\ semantics of $\prog$.  In this case, we say that the state $\qfinal$ is reachable by $\prog$.
We use  \greachof{\D,\Rel} to denote the reachability problem for a concurrent program with  the data theory ($\D, \Rel)$.
	\section{Lamport's Bakery Algorithm} \label{sec:app:lamport}
To demonstrate the practical application of our model, we use it to implement Lamport's Bakery Algorithm \cite{10.1145/361082.361093}.
Created by Leslie Lamport in 1974, it is a cornerstone solution for achieving mutual exclusion in concurrent systems. 
Picture threads as patrons entering a bakery, each is handed a unique ticket upon arrival. 
These tickets, representing the order of entry, dictate the sequence for accessing critical sections. They ensure an orderly execution flow and preventing race conditions in a critical section.

Each thread is assigned a unique number that is larger then the numbers currently assigned to other threads.
The thread possessing the lowest number is granted entry to the critical section. 
This thread may access the critical section an unbounded number of times. 
This means the assigned tickets keep increasing and thus an infinite domain is required.
Note that the algorithm does not rely on precise tickets values, 
we only need to compare the tickets to each other. This makes the protocol well suited to our program model.

The protocol contains $n$ threads where each thread $i\leq n$ is associated with two variables: The ticket number $ticket_i$ and the flag $chosen_i$ which signals whether the thread has chosen a ticket number. 
We assume $r_{TRUE}$ and $r_{FALSE}$ are initialized with different values that represent the boolean values of a flag and that $ticket_i$ is initially the same as $r_{FALSE}$ for all $i\leq n$. 

The algorithm for thread $i$ is given in \autoref{alg:bakery}. 
For the sake of simplicity and compactness we present the transition system as pseudocode. This is equivalent to a program definition since the code only accesses variables and registers using operations $\Op$ with relations $\Relinit$. The remaining instructions only affect the finite control flow and can be expressed using transitions.
It is easy to see how a corresponding program definition $\prog$ can be constructed.
We observe that this implementation of the protocol has the same asymptotic complexity as the optimal algorithm given in \cite{Ricart_Agrawala}, which uses requests and replies to maintain synchronization.

\begin{algorithm}
	\begin{algorithmic}[1]
		\STATE  $\wop {chosen_i} {r_{FALSE}}$ \label{alg:L1}  \COMMENT{Begin choosing}
		\STATE $r_i \newval$  \COMMENT{Pick random ticket}
		\FORALL{$1\leq j \leq n$} 
		\STATE $\rop {ticket_j} {r_j}$
		\IF{$(r_i<r_j)$} 
		\STATE goto line \ref{alg:L1} \COMMENT{New ticket needed.}
		\ENDIF
		\ENDFOR
				\STATE $\wop{ticket_i}{r_i}$ \COMMENT{Ticket accepted}
				\STATE $\wop {chosen_i} {r_{TRUE}}$ \COMMENT{Choosing finished}
				\FORALL{$1\leq j \leq n$}
								\STATE $\rop {chosen_j} {r_j}$\label{alg:L2}
								\IF{$(r_j\neq r_{TRUE})$} \STATE goto line \ref{alg:L2} \COMMENT{Thread $j$ is still choosing}
								\ENDIF
								\STATE $\rop {ticket_j} {r_j}$\label{alg:L3}
								\IF{$(r_j \neq r_{FALSE} ~\&~ r_j < r_i)$} \STATE goto line  \ref{alg:L3}  \COMMENT{Lower ticket $j$ found}
								\ENDIF
				\ENDFOR
				\STATE CRITICAL Section
				\STATE $r_i\define r_{FALSE}$
				\STATE goto line \ref{alg:L1} \COMMENT{Back to NON-CRITICAL}
	\end{algorithmic}
\caption{Lamport Bakery Protocol}\label{alg:bakery}
\end{algorithm}
	\section{State Reachability for \tso\ with (Dis)-Equality  Relation}\label{sec:nocontext}
	We show that the reachability problem for concurrent programs under \tso\ is undecidable when $\{=, \neq\} \subseteq \Rel$. The proof is achieved through a reduction from the state reachability problem of Lossy Channel Systems with Data (DLCS) \cite{10.1145/2933575.2934535}, which is already known to be undecidable. 
To simulate the lossy channel, we employ write buffers, as both are implemented as first-in-first-out queues. However, there are three main distinctions that must be considered:
	(i) write buffers do not contain letters, 
	(ii) write buffers are not lossy, and
	(iii) the semantics of reads differ from receives.
	
We address these distinctions as follows: 
(i) We encode the letters as variables.
(ii) We model writes being lost by avoiding to read them. 
(iii) To prevent buffer reads, we transfer the writes into a write buffer of a second thread with a different variable. We ensure that every write is accessed only once by overwriting them immediately with a different value.
\begin{theorem}\label{thm:undec}
	 \greachof{\D,\Rel} is undecidable for $\{=,\neq \} \subseteq   \Rel $. 
	\end{theorem}
	
The rest of this section is devoted to the proof of the above theorem. 
We first recall the definition of Lossy Channel Systems with Data (DLCS) \cite{10.1145/2933575.2934535}. Then, we present the reduction from state reachability problem of DLCS to \greachof{\D,\Rel}.
\begin{figure}
	\centering
\begin{equation*}
\begin{gathered}
    \infer[\text{assign}]{
		\langle q,\funV,w \rangle  \tstrans{x\define y}
		\langle q',\funV[x \leftarrow \funV(y)],w\rangle 
    }{
        \arel{q}{x \define y}{q'}\in \Delta_\lcs 
    }
    {}  
\\
{}
\infer[\text{new value}]{
	\langle q,\funV,w \rangle  \tstrans{x\newval}
	\langle q',\funV[x\leftarrow d],w\rangle 
}{
	\arel{q}{x\newval}{q'}\in \Delta_\lcs 
	& d\in \D \setminus \{\funV(y) \,|\, y \in \X_\lcs\}
}
    {}  
\\
{}
\infer[\text{equality}]{
	\langle q,\funV,w \rangle  \tstrans{x=y} \langle q',\funV,w \rangle 
}{
	\arel{q}{x=y}{q'}\in \Delta_\lcs 
	& \funV(x)=\funV(y)
}
    {}  
\\
{}
\infer[\text{disequality}]{
	\langle q,\funV,w \rangle  \tstrans{x\neq y} \langle q',\funV,w \rangle 
}{
	\arel{q}{x\neq y}{q'}\in \Delta_\lcs 
	& \funV(x)\neq \funV(y)
}
    {}  
\\
{}
\infer[\text{send}]{
	\langle q,\funV,w \rangle  \tstrans{!\langle a,x\rangle} \langle q',\funV,(a,\funV(x)).w \rangle 
}{
	\arel{q}{!\langle a,x\rangle}{q'}\in \Delta_\lcs 
	& 
}
    {}  
\\
{}
\infer[\text{receive}]{
	\langle q,\funV,w . (a,d) \rangle  \tstrans{?\langle a,x\rangle} \langle q',\funV[x\leftarrow d],w \rangle 
}{
	\arel{q}{?\langle a,x\rangle}{q'}\in \Delta_\lcs 
	& 
}
    {}  
\\
{}
\infer[\text{lossiness}]{
	\langle q,\funV,w\rangle  \tstrans{loss} \langle q,\funV,w' \rangle
}{
	& w' \leq w
}
\\
\end{gathered}
\end{equation*}
\caption{The transition relation of DLCS}
\label{fig:rules_dlcs}
\end{figure}	
\paragraph*{Lossy Channel Systems with Data}
A DLCS $\lcs= \langle \Q_\lcs , \X_\lcs, \Sigma_\lcs, \Delta_\lcs, q_{\init}\rangle $ consists of a finite set of states $\Q_\lcs$, a finite number of variables $\X_\lcs$ ranging over an infinite domain $\D$, a finite channel alphabet $\Sigma_\lcs $, $q_{\init} \in \Q$  is the initial state, and a finite set of transitions $\Delta_\lcs$. The set $\Delta_\lcs$ of transitions is a subset of $  \Q_\lcs \times \Op_\lcs \times \Q_\lcs$. Let $x, y \in \X_\lcs$. The set $\Op_\lcs$ consists of the following operations (1) $x\define y$ which assigns the value of $y$ to $x$, (2) $x\newval$, which assigns a fresh value from $\D$ that is different from the existing values of all variables\footnote{This differs from the $\circledast$ in \tso\  where the value $d\in \D$ assigned by the operation $x\newval$ can be anything.},  $(3)$ ${x}={y}$ ($x\neq y$) which compares the value of variables  $x$ and $y$, $(4)$ $!\langle a,x\rangle$ which appends letter $a \in \Sigma_\lcs$ together with the value of $x$ to the channel, $(5)$ $?\langle a,x\rangle$ which deletes the head of the channel $\langle a,d\rangle$ and stores the value $d$ in $x$, and (6) $loss$ which removes elements in the channel.

A configuration $\conf$ of DLCS is defined by the tuple $\langle q, \funV, w\rangle$ where $q \in \Q_\lcs$ is the current state, $\funV: \X_\lcs \rightarrow \D$ is the current valuation of the variables, and $w \in (\Sigma \times \D)^*$ is the content of the lossy channel. 
The system is lossy, which means any element in the channel may disappear anytime.
The initial configuration $\conf_{\init}$ of $\lcs$ is defined by $(q_{\init},\funV_{\init},\epsilon)$ where  $\funV_{\init}(x)=0$ for all $x \in \X_{\lcs}$. The transition relation of DLCS is given in Figure \ref{fig:rules_dlcs}.

The state reachability problem for $\lcs$ asks whether, for a given final state $\qfinal \in \Q$, there is a reachable configuration $\conf$ of the form $\conf=\langle \qfinal, \funV, w \rangle$. In this case, we say that the state $\qfinal$ is reachable by $\lcs$.

\begin{theorem}[\cite{10.1145/2933575.2934535}]
	The state reachability problem for DLCS is undecidable.
\end{theorem}
\begin{figure}
\centering
 \usetikzlibrary {arrows.meta} 
\usetikzlibrary{arrows}
\usetikzlibrary{shapes}
\tikzset{every picture/.style={>=stealth,bend angle=20}}
\tikzset{every label/.style={font=\scriptsize}}
\tikzset{every node/.style={font=\scriptsize}}
\tikzset{play/.style={circle,draw,minimum size=#1}}
\tikzset{play/.default=0.625cm}

\tikzset{pla/.style={circle,draw,fill=red!30,minimum size=#1}}
\tikzset{pla/.default=0.5cm}
\tikzset{plb/.style={circle,draw,fill=yellow!30,minimum size=#1}}
\tikzset{plb/.default=0.5cm}
\tikzset{player0/.style={circle,draw,fill=green!30,minimum size=#1}}
\tikzset{player0/.default=0.5cm}

\tikzset{prob/.style={diamond,draw,fill=blue!20,minimum size=#1}}
\tikzset{prob/.default=0.7cm}
\tikzset{end/.style={rectangle,draw,minimum size=#1}}
\tikzset{end/.default=0.55cm}

\begin{tikzpicture}[label distance=10mm]

\node (q1)  [pla] {$q$};
\node (2) [pla, right = of q1] {};
\node (q3) [pla, right = of 2] {};
\node (4) [pla, right = of q3] {};
\node (5) [pla, right = of 4] {};
\node (6)[pla, right = of 5] {};
\node (7)[pla, right = of 6] {$q'$};
\draw[->] (q1) to node[above=1mm]{$\rtmp :=\circledast$}  (2);
\draw[->] (2) to node[above=1mm]{$\rtmp \neq r_\$ $}  (q3);
\draw[->] (q3) to node[above=1mm] (test) {$\rtmp  \neq r_{x_1}$}  (4);
\draw[ dotted,->,thick] ([xshift=5pt]4.east) to node[above=1mm]{}  ([xshift=-5pt]5.west);
\draw[->] (5) to node[above=1mm] {$\rtmp \neq r_{x_n}$}  (6);
\draw[->] (6) to node[above=1mm] {$r_x \define \rtmp $}  (7);
\node (caption) [above=0.4cm of 4] {$Gad^t_{x:=\circledast} : (q, x:=\circledast, q')$};

\end{tikzpicture}
\vspace*{0.5cm}

\begin{tikzpicture}[label distance=10mm]

\node (1)  [pla] {$q$};
\node (2)  [pla, right = of 1] {};
\node (3) [pla, right = of 2] {};
\node (4)  [pla, right = of 3] {};
\node (5) [pla, right = of 4] {$q'$};
\node (caption) [above= 0.4cm of 3] {$Gad^t_{?\langle a, x\rangle} : (q, ?\langle a, x\rangle, q')$};

\draw[->] (1) to node[above=1mm]{$\rop{y_a}{r_x}$}  (2);
\draw[->] (2) to node[above=1mm]{$r_x \neq r_\$ $}  (3);
\draw[->] (3) to node[above=1mm]{$\rop{y_a}{\rtmp }$}  (4);
\draw[->] (4) to node[above=1mm]{$\rtmp =r_\$ $}  (5);
\end{tikzpicture}
\vspace*{0.2cm}

\begin{tikzpicture}[label distance=10mm]

\node (1) [pla] {$q$};
\node (2) [pla, right = of 1] {};
\node (3right)  [pla, right = of 2] {$q'$};

\draw[->] (1) to node[above=1mm]{$\wop{x_a}{r_x}$}  (2);
\draw[->] (2) to node[above=1mm]{$\wop{x_a}{r_\$}$}  (3right);
\node (caption) [above=0.5cm of 2] 
{$Gad^t_{!\langle a, x\rangle} : (q, !\langle a, x\rangle, q')$};

\node (6)  [plb, right =2cm of 3right] {};
\node (5) [plb, right = of 6] {};
\node (4)  [plb, right = of 5] {};

\node (1) [plb, above= of 6] {$q_{\ch}$};
\node (2) [plb, right = of 1] {};
\node (3) [plb, right = of 2] {};

\draw[->] (1) to node[above=1mm]{$\rop{x_a}{\rchtmp}$}  (2);
\draw[->] (2) to node[above=1mm]{$\rchtmp \neq \rchnew $}  (3);
\draw[->] (3) to node[left]{$\wop{y_a}{\rchtmp}$}  (4);
\draw[->] (4) to node[below=1mm]{$\rop{x_a}{\rchtmp}$}  (5);
\draw[->] (5) to node[below=1mm]{$\rchtmp=\rchnew$}  (6);
\draw[->] (6) to node[left]{$\wop{y_a}{\rchtmp}$}  (1);


\end{tikzpicture}

\caption{$\prog(\lcs)$ with threads $t$ (pink states) and $t_{\ch}$ (yellow states).}
\label{fig:undec}
\end{figure}
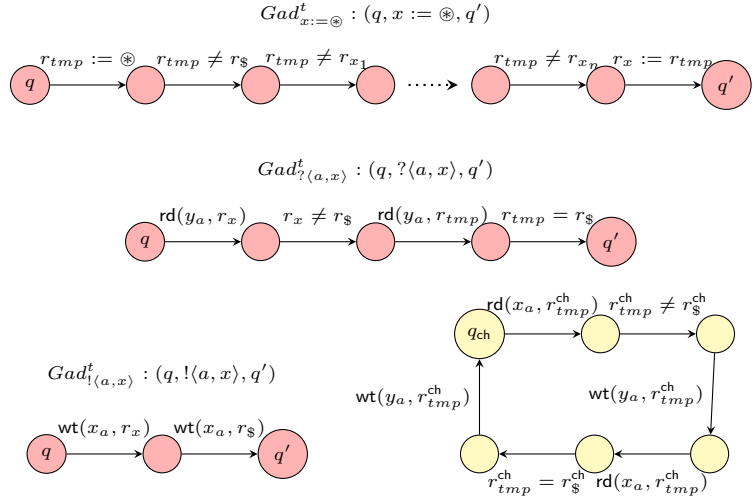	
\newpage

\paragraph*{Reduction from DLCS reachability} Given a DLCS $\lcs= \langle \Q_\lcs, \X_\lcs, \Sigma_\lcs, \Delta_\lcs, q_{\init}\rangle $ over  data domain $\D$ with  $\X_\lcs=\{ x_1\ldots x_n \}$, we reduce the state reachability of $\lcs$ to the reachability problem \greachof{\D, \{=,\neq\}} of a concurrent program 
$\prog (\lcs)$, with two threads $t, t_{\ch}$.
 The thread $t$ simulates the operations of $\lcs$, while thread $t_{\ch}$ simulates the 
	lossy channel of $\lcs$ using its write buffer. Let $\R_t=\{\rnew ,\rtmp \} \cup \{r_x \mid x \in \X_\lcs\}$,  $\R_{t_{\ch}}=\{\rchnew, \rchtmp \}$ be the local registers of threads $t$ and $t_{\ch}$. 
		Corresponding to each  $x \in \X_\lcs$, we have the register $r_x$ in thread $t$, which stores the current values of $x$. 
	Registers $\rtmp$ and $\rchtmp$ are used to temporarily store certain values. 
	The shared variables  of $\prog(\lcs)$ are   $\X=\{x_a, y_a \mid a \in \Sigma_\lcs \}$, they help in simulating the behavior of the lossy channel of $\lcs$. 
	  
\noindent \emph{Simulating the DLCS}. 
The transitions of $\prog(\lcs)$ are sketched in \autoref{fig:undec}. The initialization of the program is omitted in the figure and goes as follows. The thread $t_{\ch}$ starts by assigning a non-deterministic value (say $\$$) to the register  $\rchnew$ (i.e., $\rchnew \newval$), then checks that the new value $\$$ is different from $0$ (i.e., by checking that $\rchnew \neq \rchtmp$), and finally performs an atomic read write operation $ \arw{x} {\rchtmp}{\rchnew}$ on each  variable  $x \in \X$. The thread $t$  starts by reading the value of each shared variable $x \in \X$ (i.e., performing $\rop x {\rnew}$) and checks if its value is different from $0$ (i.e., $\rnew \neq  \rtmp$). At the end of this initialization phase,  all the shared variables have the new value $\$$, the registers $\rtmp$ and $\rchtmp$ have the value $0$ and the registers $\rnew$ and $\rchnew$ have the value $\$$. The current state of thread $t$ is the initial state $q_{\init}$ of $\lcs$ while the thread $t_{\ch}$ is in a state $q_{\ch}$.

Every transition $\arel{q}{x \define y}{q'}\in \Delta_\lcs $ is simulated in $\prog (\lcs)$ by threat $t$ with a gadget—a sequence of transitions that starts in $q$ and ends in $q'$.
The transitions $(q, x\define y, q')$,  $(q,  x=y, q')$ and  $(q, x \neq y, q')$ in the DLCS are simulated by the thread $t$
as gadgets with single transitions $(q,r_x\define r_y,q'), (q,r_x=r_y,q')$ and $(q,r_x \neq r_y,q')$, respectively. We omit their description in \autoref{fig:undec}.

To simulate $x\newval$, we load the new value in register $\rtmp $ and ensure that it is different from the values in registers $\rnew$ and $r_{x_1} \ldots r_{x_n}$. 
This is depicted by the gadget $Gad^t_{x\newval}$ in thread $t$.
The send operation $!\langle a, x\rangle$ in the DLCS is simulated by the gadget $Gad^t_{!\langle a, x\rangle}$.
In the DLCS, the send appends the letter $a$ and the value of $x$ to the channel. This is simulated by the write $\wop{x_a}{r_x}$, thereby 
  appending  $(x_a, val(r_x))$ to the buffer of $t$. To simulate reads of the DLCS, we first make note of a crucial difference in the way reads happen in DLCS and \tso. In DLCS, a read happens from the head of the channel, and the head is deleted immediately after the read. In \tso\ however, we can read from the latest write in the shared memory multiple times. In order to simulate the ``read once'' policy of the DLCS, we follow each $\wop{x_a}{r_x}$ with another write $\wop{x_a}{\rnew}$.  
   
Thread $t_{\ch}$ is a loop from the state $q_{\ch}$ which continuously reads from $x_a$ a value from a simulated send followed by the separator $\$ $.
It copies these values to $y_a$ using local register $\rchtmp$.
The first time it reads from $x_a$, it reads the value $d$ of $x$ from a simulated send $!\langle a, x\rangle$. It ensures that this is not the $\$$ symbol  ($\rchtmp \neq \rchnew$), and  writes this value from $\rchtmp$ into variable $y_a$, thus 
appending $(y_a, d)$ in the buffer of $t_{\ch}$. It then reads again the value  of $x_a$  into $\rchtmp$. This time, it makes sure	 
to read $\$ $ with the check $\rchtmp=\rchnew$. 
The receive $?\langle a, x\rangle$ of the DLCS is simulated by $Gad^t_{?\langle a, x\rangle}$. First, we read from $y_a$ and store it in $r_x$, ensuring this value $d$ is not $\$$. 
Then, we read $\$$ from $y_a$. This ensures that the earlier value $d$ 
is overwritten in the memory and is not read twice. 

A loss in the channel of the DLCS results in losing some messages $\langle a, d\rangle$. This is accounted for in $\prog_\lcs$  in two ways. 
Thread $t_{\ch}$ may not pass on a value written from $x_a$ to $y_a$ since the loop may not execute for every value.
Thread $t$ may not read a value written by $t_{\ch}$ in $y_a$ since it was already overwritten by some later writes.

	\begin{lemma}\label{lemma:undec}
	 The state  $q_{final}$ is reachable by $\lcs$ if and only if $q_{final}$ is reachable by $\prog(\lcs)$. 
	\end{lemma}
	The formal proof is in \autoref{sec:app_undec}. 
	\autoref{thm:undec} extends to any set of relations that we can use to simulate equality and disequality. For instance $\leq,\nleq\in \Rel$.
	\section{Context Bounded Analysis}
In the light of this undecidability, we turn our attention to a variant of the reachability problem which is tractable.  
We study context bounded runs, an under-approximation of the program behavior that limits the possible interactions between processes.
A run consists of a number of \emph{contexts}.
A context is a sequence of steps where only a certain fixed thread $t$ is \emph{active}. 
We say that $\run \in \pmc (k)$ if and only if there is a partitioning $\run=\run_1 \ldots \run_k$ such that for all contexts $i\leq k$ there is an active thread $t_i\in \T$ such that 
only the active thread  updates the memory and performs operations: If $\conf \tstrans{\ell} \conf'\in \run_i$, then $\ell \in \{t_i\} \times (\Op \cup \{ u\} )$.
\medskip

In the following, we show \pspace\ completeness of \reachof{\pmc(k)}{\D,\Relord} for relations such as (dis) equality, ``greater than'' or even ``greater by at least $n$'' for $n \in \mathbb{N}$ (see Theorem \ref{cor:pmctotalorder}).
Our approach begins with a  proof of \pspace\ hardness 
through a reduction from the non-emptiness problem of the intersection of regular languages \cite{4567949}.

Next, we demonstrate \pspace\ membership by reducing the problem to state reachability of a finite transition system which we solve in polynomial space.
This reduction faces challenges from two main sources, namely, (i) the unbounded size of the write buffers, and (ii) the infinite data domain $\D$.
 In this section, we show how to construct a finite transition system while preserving state reachability in two key steps.
 
 Following \cite{ABP2011}, we first perform a buffer abstraction.
 An in-depth analysis of the \tso\ semantics within context bounded runs reveals a critical insight: 
 Even though the buffer may contain an unbounded number of writes, only a bounded number of these writes can be read later on. 
 This allows us to non-deterministically identify and store the necessary writes using variables.
 
 Finally, we implement a domain abstraction. 
 A popular approach is to abstract the values into equivalence classes based on the supported relations. This reveals our next challenge: (iii) the set of relations $\Relord$ is infinite. 
 We conduct an analysis of the reachable configurations and discover the following: 
 If a configuration is reachable, then any configuration that is the same except with greater distances between differing values is reachable as well. 
 It follows that, for control state reachability, the abstraction does not require the precise distances between variables; their relative order is sufficient.

\subsection{Lower-bound}
We establish \pspace\ hardness by polynomially reducing the problem of checking non-emptiness of the intersection of regular languages to \reachof{\pmc(k)}{\D,\Relord}.
Given a set of finite automata $\fsm_1 \ldots \fsm_n$ with   $\fsm_i= \langle \Q_i, \Delta_i, q^\init_i \Q^F_{i} \rangle$, where $\Delta_i\subseteq \Q_i \times \Sigma \times \Q_i,\; q_i^\init \in \Q_i$, and $\Q^F_{i}\subseteq \Q_i$ for $i\leq n$, the problem asks whether there is a word $w\in \Sigma^*$ that is accepted by each automaton $\fsm_i$ with $i\leq n$.
This is known to be \pspace\ hard\cite{4567949}.

We construct a program $\prog ( \fsm_1 \ldots \fsm_n )$  that consists of a single thread and reaches a state $\qfinal$ if and only if there is such a word.
The idea of the construction is that we assign each state $q_i\in \Q_i$ a unique value stored in a register $r_{q_i}$ and we store the value of the current state of each automaton $\fsm_i$ in a register $r_i$. 
To begin, we ensure that the current states are the initial ones. This means $r_i=r_{q^\init_i}$ holds for each $i\leq n$.
Then, we choose a letter $a\in \Sigma$ and simulate some transition $q_i\tstrans{a} q'_i \in \Delta_i$ for each automaton. This is done by ensuring that the current state is $q_i$ with $r_i=r_{q_i}$ and then updating the current state with $r_i\define r_{q'_i}$. We repeat this step until each current state is a final state. At this point, we know we have simulated runs for each automaton that accept the same word and we reach $q_\final$.

The formal definition of the construction as well as the proof of correctness is given in \autoref{sec:app:pspacehard}.
This is a polynomial reduction of non-emptiness of the intersection of regular languages to \reachof{\pmc(k)}{\D,\Relord}. Observe that we only need test for equality and disequality. The disequalitiy checks are necessary to ensure that each register $r_{q_i}$ has been assigned a different value. 

\begin{theorem}\label{thm:pmchardness}
	\reachof{\pmc(k)}{\D,\Relord} is  \pspace\ hard.
\end{theorem}

\subsection{\pspace\ Upper-bound}
Assume that we are given a program $\prog$ and a context bound $k$. As an intermediary step towards finite state space we construct a finite state machine $\scbab$ with variables, over the infinite data domain $\D$.  The name $\scbabname$ stands for \textit{abstract buffer} as it abstracts from the unbounded write buffers using a finite number of variables.  
We show that $\scbab$ is state reachability equivalent with the \tso\ semantics of $\prog$ bound by $\pmc(k)$.

While abstracting away the buffers, the main challenge is to simulate read operations. Recall from \autoref{sec:tso} that each read operation in  a thread accesses either a write from its own buffer or from the shared memory.
A buffer read always reads from the threads latest write on the same variable.
Since only the active thread may interact with the memory during the context, we can assume w.l.o.g. that all memory updates occur at the end of a context. 
This means a memory read accesses the last write on the same variable that updated the memory in an earlier context, and hence 
we do not need to store the whole buffer content. For \emph{memory reads},
 we need the latest writes leaving the buffer at the end of each context for each variable. 
For \emph{buffer reads}, we only require the latest writes on each variable that are issued by each thread.

\paragraph*{Construction of the abstract machine} The abstract machine  $\scbab$ is defined by the tuple $\langle \Q_{\scbabname},\X_{\scbabname},\Delta_{\scbabname}, q_{\init}^\scbabname\rangle$ where $ \Q_{\scbabname}$ is the finite set of states, $\X_{\scbabname}$ is the finite set of variables, $\Delta_{\scbabname}$ is the transition relation, and $q_{\init}^\scbabname$ is the initial state. 
A control state $q_\scbabname \in \Q_{\scbabname}$ is a tuple $(\funQ,act,j,c,u)$ where: (i) 
the current state of every thread is stored using function $\funQ : \T \rightarrow \Q$; 
(ii) function $act: \{1\ldots k \} \rightarrow \T$ assigns to each context an active thread; (iii)
 the current context is stored in variable $j\in \{1\ldots k \}$;
(iv) the function $c: \X  \times \T \rightarrow \{0, 1\ldots k \}$ assigns to each variable $x\in \X$ and thread $t\in \T$, the (future) context $j'$ in which the latest write on $x$ will leave the write buffer of $t$. 
 This determines when $t$ can access the shared memory on that variable again; 
 and (v) function $u: \{1\ldots k \} \rightarrow 2^\X$ assigns each context $j$ the set of variables that are updated during $j$.
Additionally, we will introduce some helper states with the transitions relation. 
We omit them from the definition of  $\Q_{\scbabname}$.
The initial state $q_{\init}^\scbabname$ is such a helper state.

The set of variables $\X_{\scbabname}$ contains:
(i) the set of variables $\X$ in $\prog$, (ii) the set of registers $\R$, 
	(iii) for each each context $j\leq k$ and each variable $x\in \X $, we introduce a variable $x_j$, which  stores the value of the last write on $x$ that leaves the write buffer in context $j$, (iv) for each thread $t$ and each variable $x\in \X$, we introduce a variable $x_t$ which stores the value of the newest write of $t$ on $x$ that is still in the buffer of $t$. Notice that this is the write that $t$ accesses when reading $x$ (if such a write exists).  

	We define the transition relation $\Delta_{\scbabname}$ in \autoref{fig:rules_scbab}.
	\begin{figure*}[tb]
	\includegraphics[width=\textwidth]{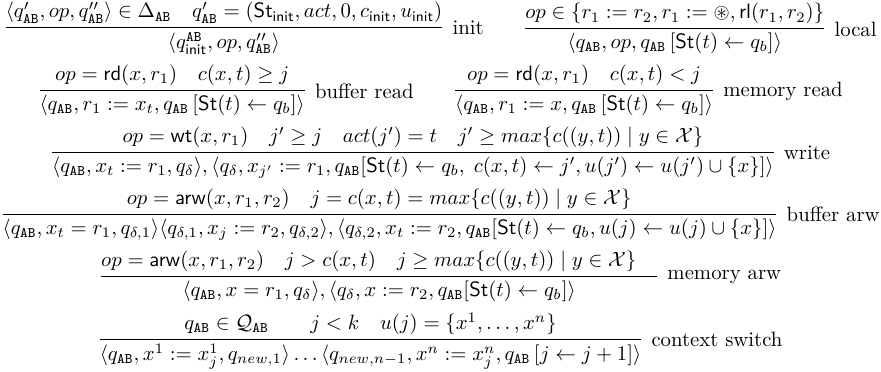}
	\caption{The transition relation $\Delta_{\scbabname}$ of $\scbab$. Let $\delta=\arel{q_a} {op} {q_b}\in \Delta_{t}$ and $q_\scbabname=(\funQ,act,j,c,u)$ with  $\funQ(t)=q_a$ and $act(j)=t$.}
	\label{fig:rules_scbab}
\end{figure*}
	Let $c_{\init}(x,t)=0$ for all $x \in \X$ and $t \in \T$, and $u_{\init}(i)=\emptyset$ for all $i \in \{1,\ldots,k\}$. 
	The outgoing transitions of state $q^{\scbabname}_{\init}$ 
	are the outgoing transitions of $(\funQ_{\init},act,0,c_{\init},u_{\init})$ for every possible function $act$.
	This means the construction guesses a function $act$ and behaves as if the other elements in the tuple have the initial values.
	Local transitions are adapted in a straightforward manner.
	A read on $x$ from the buffer occurs if there is a write on $x$ in the buffer. This means the latest write on $x$ leaves the buffer in a context $c(x,t)$ after (or in) the current context $j$.
	In such a case, we access $x_t$ which holds the latest write on $x$ in the buffer of $t$. 
	If there is no such write on $x$ in the buffer, i.e. $c(x,t)<j$ holds, then the read fetches the value of $x$ from the shared memory.
	
	A write operation on $x$ overwrites the latest entry in the write buffer on that variable $x_t$ and determines a future (or current) context $j'$ with $j'\geq j$ in which it leaves the buffer.
	This is recorded in the variable $x_{j'}$ and $x$ is added to the set $u(j')$ which holds the variables that are updated in context $j'$.
	Note that $j'$ cannot be smaller then any other context in which a write on a variable $y$ leaves the buffer of $t$. This information is obtained from the function $c$. Also, $j'$ must be a context in which $t$ is active.
	
	At any time, the run can switch from a context $j$ with $j<k$ to $j+1$. 
	Let $u(j)=\{ x^1\ldots x^n \}$. These are the variables that are updated during context $j$. 	The values of the last updates on these variables in the context, stored in $x^1_j\ldots x^n_j$, are written to the corresponding variables in the shared memory. 
Since $\scbab$ only performs memory updates at the end of a context, an atomic read write $\arw{x}{r_1}{r_2}$ requires that the current buffer content leaves the buffer in the current context. 
This is ensured by using the condition $  j\geq  max\{ c((y,t )) \mid y\in \X \}$. 
If there is a write on $x$ in the buffer of $t$, then  $j=c(x,t)$. This is covered by the \emph{buffer arw} rule in \autoref{fig:rules_scbab}. Here, the current value of $x$ is stored in $x_t$, so we first check that it equals $r_1$ and update $x_t$ as well as $x_j$ with $r_2$.
If $j>c(x,t)$ holds, then there is no write on $x$ in the buffer of $t$ (\emph{memory arw} rule) and we compare the value of $x$ in the shared memory with $r_1$ and update it to $r_2$.

 A configuration $\conf=(q_\scbabname, \funX)$ in the induced LTS of $\scbab$ consists of a state  $q_\scbabname \in \Q_\scbabname$ along with a variable assignment $\funX$.
Let $\conf_{\init}=(q^\scbabname_{\init}, \funX_{\init})$ be the initial configuration of  $\scbab$. 
 Given the transitions $\Delta_{\scbabname}$, 
we can define the transitions in the induced LTS 
in a straightforward manner.
A  state $\qfinal \in \Q_t$ of thread $t$ is said to be reachable by $\scbab$ if and only if there is a reachable configuration of the form $((\funQ,act,j,c,u), \funX)$ such that $\funQ(t)=\qfinal$ holds. 

\medskip

\begin{lemma}\label{lem:scbab}
	A state of $\prog$ is reachable under \tso\ by a run $\run \in \pmc(k)$ if and only if it is reachable by $\scbab$.
\end{lemma}

The proof of \autoref{lem:scbab} is given in \autoref{sec:app_scbab}. Next, we abstract away the infinite data domain from $\scbab$. 
We remove this last source of infinity by constructing a finite state machine $\relAB$ from $\scbab$. 

\subsubsection{Domain Abstraction}
We use domain abstraction to solve \reachof{\pmc(k)}{\D,\Relord} by reducing state reachability of $\scbab$ to reachability of a finite state machine.
We introduce the set of relations $\Relinit=\{ =, \neq , < \}$.
To abstract away the infinite data domain, we abstract from the exact values of the variables. Instead of storing actual values, we store which relations from $\Relinit$ holds between which pairs of variables, which is finite information. This way, we reduce the infinite domain $\D$ to the finite Boolean domain $\B$.
For example, $( q_\scbabname ,  x=y )$ is an abstraction of a configuration $(q_\scbabname,  \funX(x)=1,\funX(y)=1 )$.
Given a variable assignment $\funX$ and a relation $\rlvar$, we define $\rlvar_\funX  (x,y) \define \rlvar (\funX (x),\funX (y))$.
Any variable assignment $\funX$ induces a set of relations $\Rel_\funX=\{ \rlvar_\funX \mid \rlvar \in \Relinit \}$ over the variables $\X_{\scbabname}$. 
When considering multiple sets of relations we denote a relation $\rlvar \in \Rel$ as $\rlvar_\Rel$.
For a variable assignment $\funX$, we say set of relations $\Rel$ over variables  is consistent with $\funX$ if $\Rel=\Rel_{\funX}$. 

Given $\scbab=\langle \mathcal{Q_{\scbabname},X_{\scbabname}},\Delta_{\scbabname}, q_{\init}^{\scbabname}\rangle $, we now construct the finite state machine $\relAB=\langle \Q,\Delta,q_{\init}\rangle $ as follows: 
$\Q \define \Q_{\scbabname} \times \{ \rlvar_{\X_{\scbabname}} : \X_{\scbabname} \times \X_{\scbabname} \rightarrow \B \mid \rlvar \in \Relinit \}$. 
We abstract from a variable assignment by storing in the states which relations are satisfied. 
The initial state is $q_\init=(q^\scbabname_\init , \Rel_{\funXinit})$.
We define the transitions of  $\relAB$ in \autoref{fig:rules_relAB}.
\begin{figure*}[tb]
	\centering
\includegraphics[width=\textwidth]{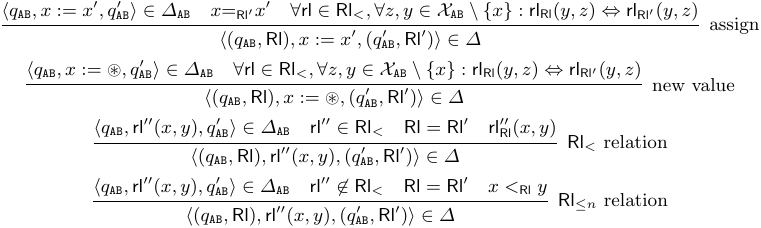}
	\caption{The transition relation of  $\relAB$. Sets $\Rel$ and $\Rel'$ satisfy (i) equality is an equivalence relation;
		(ii) disequality holds iff equality does not hold;
		(iii) $"<"$ is a total order on variables that are not equal.}
	\label{fig:rules_relAB}
\end{figure*}
We construct the transitions such that they abstract from the transitions of the LTS induced by the semantics of $\scbab$.
Where the semantics on transitions of $\scbab$ require that certain values in the configurations before and after the operation are the same, the transitions of $\relAB$ only require that the relations between variables before and after the relation are the same. For instance, the assign rule for operation $x\define x'$ requires that $\Rel$ and $\Rel'$ are the same for all variables except $x$ and $x=_{\Rel'} x'$ must hold after the operation.
Conditions (i)-(iii) in \autoref{fig:rules_relAB} reflect the properties of $\Relinit$ on values. They ensure that $\Rel$ and $\Rel'$ have consistent variable assignments.
Note that for any operation $<_n$ (or $\leq_n$), we soften the condition to $x<_\Rel y$. We will show that this still results in an abstraction precise enough to be state reachability equivalent.

Since $\relAB$ is a finite state machine, it induces the obvious LTS where a configuration consists of a state.
  The following lemma shows that the construction is indeed an abstraction of $\scbab$. We assume  $\prog$ uses $\Relord$.
 \begin{lemma}\label{lem:abstract}
 	If $q_\scbabname$ is reachable by $\scbab$, then a state $(q_\scbabname,\Rel)$ is reachable by  $\relAB$.
 \end{lemma}
 
 \medskip

\begin{proof}
	Assume $\langle (q_\scbabname, \funX )\tstrans{op} (q'_\scbabname, \funX' ) \rangle$.
	We argue that $\arel{(q_\scbabname, \Rel_{\funX})} {op} {(q'_\scbabname, \Rel_{\funX'} )}\in \Delta$ holds as well. The lemma follows immediately.
	We show this for operation $x\newval$. For all other operations, the proof is analogue and we omit it.
	
	It follows from the semantics of $x\newval$, that $\funX(y)=\funX'(y)$ for any $y \in \X_\scbabname\setminus\{x\}$ holds. 
	This means $\Rel_{\funX}$ and $\Rel_{\funX'}$ satisfy the new value rule. 
	The equality relations in  $\Rel_{\funX}$ and $\Rel_{\funX'}$ are consistent with the equality relations on values of $\funX$ and $\funX'$. The equality relation given by the values is an equivalence relation and thus Condition (i) is satisfied. Similarly, Condition (ii) is satisfied since values are obviously not equal if and only if they are not related by equality. Condition (iii) is satisfied since relation $<$ on values forms a total order.
	All conditions are satisfied. This means $\arel{(q_\scbabname, \Rel_{\funX})} {x\newval} {(q'_\scbabname, \Rel_{\funX'} )}\in \Delta$.
\end{proof}
\medskip

\begin{lemma}\label{thm:relAB}
	If a state $(q_\scbabname,\Rel)$ is reachable by $\relAB$, then $q_\scbabname$ is reachable by $\scbab$.
\end{lemma}
We prove this by performing an induction over runs of $\relAB$ and constructing equivalent runs of $\scbab$. In order to do this, we construct configurations with consistent variable assignments. 
The main challenge is that these variable assignments may not have large enough distances between the values. Take the operation $x <_n y$, for instance. Here, $\relAB$ only requires $x < y$. Note that any value other than $0$ was created by an $x\newval$ operation. We can modify a run so that some of these operations assign larger values. This way, we can increase the distances of variable assignments of reachable configurations without changing their consistency with respect to relations. The formal proof of this is given in \autoref{sec:app:relAB}.

\begin{theorem}\label{cor:pmctotalorder}
	\reachof{\pmc(k)}{\D,\Relord} is \pspace\ complete.
\end{theorem}
\begin{proof}
	While $\Relord$ is an infinite set, $\Relinit$ has only $3$ relations.
	This means $\relAB$ is a finite transition system where state reachability is decidable.	
	According to \autoref{lem:scbab}, \autoref{lem:abstract} and \autoref{thm:relAB}, deciding state reachability of $\relAB$ is equivalent to solving \reachof{\pmc(k)}{\Relord}.
	
	 We non-deterministically solve the state reachability of $\relAB$ by guessing a run that is length-bounded by the size of the state space and checking whether it reaches $\qfinal$. 
	 We store the current state $((\funQ,act,j,c,u), \Rel)$ together with a binary encoding of the current length of the run. Note that the state only requires polynomial space. The number of states of $\relAB$ is exponential in the program size as well as $k$, which means the binary encoding also requires polynomial space.
	
	We extend the run by choosing to either perform a context switch or an operation. 
	We begin with the initial state $q^\scbabname_\init$, which is a special case since we first need to guess a function $act$ according to the init rule in \autoref{fig:rules_scbab}. 
	To perform an operation, we look at the current state of the active thread $\funQ(act(j))$, pick an outgoing transition from the program, and update the state according to the corresponding rules given in \autoref{fig:rules_scbab} and \autoref{fig:rules_relAB}.
	
	We illustrate this on the new-value operation. Assume we pick the outgoing transition $\arel{q_a} {x\newval} {q_b}\in \Delta_{act(j)}$. In this case, we update the state according to the local rule in \autoref{fig:rules_scbab}. Then we update the set $\Rel$ according to the new-value rule in \autoref{fig:rules_relAB}. We leave all relations that do not include $x$ unchanged, and we non-deterministically choose $x$ to be either equal to some variable, or to be between two other adjacent variables, or to be the largest or smallest variable. We update the relations to $x$ accordingly. For any other operation, the changes to $\Rel$ are uniquely determined. For writes, we additionally need to non-deterministically pick some future context $j'$ of the update according to the write rule in \autoref{fig:rules_scbab}. In the case of a context switch, we perform a series of variable assignments according to the context switch rule.
	
	Note that we do not explicitly construct the entire $\relAB$ transition system; the program and the rules given in \autoref{fig:rules_scbab} and \autoref{fig:rules_relAB} are sufficient to guess a run. Each step can be performed in polynomial space. Once  $\funQ(act(j))=\qfinal$ holds, we know $\qfinal$ is reachable. The complexity of this process is in \pspace. According to \autoref{thm:pmchardness}, the problem is \pspace\ hard as well.
\end{proof}

\section{Conclusion}
We examined safety verification of concurrent programs running under \tso\ that operate on variables ranging over an infinite domain.
We have shown that this is undecidable even if the program can only check the variables for equality and non-equality. 
We studied a context bounded variant of the problem as well.
Here, we solved the problem for programs using relations in $\Relord$ and showed that it is \pspace\ complete.

As future work,  we plan to examine more expressive under-approximations of the program behaviour than the presented context bounded analysis and how these under-approximations affect decidability and complexity of the problem.
We also intend to explore the problem for additional relations and/or operations a program may perform.

\bibliography{main}
\appendix
\section{Proof of \autoref{thm:undec}}\label{sec:app_undec}
	Since both the write buffers and lossy channels can be emptied without changing the state at the end of a run, it is sufficient to only examine reachable configurations $\conf'$ with empty buffers and empty channels for the reachability problems for \tso\ and lossy channels. 
	
	We now examine a run of $\lcs$ and we construct an equivalent run of the program. 
	It does not matter for reachability when an element in the channel gets lost only if it gets lost. This means we can assume w.l.o.g. that any channel content that gets lost is lost immediately in the run, meaning the only place where a message gets lost is in the transition it gets send. 
	It is easy to see that we can simulate a run in the lossy channel system as follows:
	For any configuration $\conf$ of $\lcs$ the corresponding configuration $\conf_{TSO}$ of $\prog$ is as follows:
	The state of $t$ is the state of  $\lcs$, the state of $t_{ch}$ is $q_{ch}$. 
	The write buffer of $t_{ch}$ is consistent with the channel $ch$. 
	The write buffer of $t$ is empty. The registers $\X$ are equal to the corresponding variables of $\conf$ and everything else is $\rnew$. 
	
	We replace any transition $\conf \xrightarrow[]{ch!\langle a,x\rangle } \conf'$ where the message is lost with the following sequence of transitions where the value is written and leaves the write buffer but it is never read:
	$ \conf_{TSO}\xrightarrow {t,\wop {x_a} {r_x}}. \xrightarrow{t,\wop{x_a}{\rnew}} .  \xrightarrow{t,u }. \xrightarrow{t,u } .\conf'_{TSO}$. 
	We replace any transition $ \conf \xrightarrow[]{ch!\langle a,x\rangle } \conf'$ where the message is not lost with the following sequence of transitions where the value is read by $t_{ch}$ and put into its write buffer
	$\conf_{TSO}\xrightarrow{ t,\wop {x_a} {r_x}}. \xrightarrow {t,\wop{x_a}{\rnew}}. \xrightarrow{ t,u } $.
	$\xrightarrow{ t_{ch},\rop{x_a}{\rtmp}}. \xrightarrow {t_{ch}, \rtmp \neq \rnew}. \xrightarrow {t_{ch},  \wop{y_a}{\rtmp}} . \xrightarrow {t_{ch},\wop{y_a}{\rnew}}. \xrightarrow {t,u } .\xrightarrow {t_{ch}, \rop{x_a}{\rtmp}}. \xrightarrow {t_{ch} ,  \rtmp = \rnew}\conf'_{TSO}  $.
	Note that $t_{ch}$ always returns to $q_{ch}$ and the buffer of $t$ remains empty after any such sequence.
	
	We replace any transition $\conf \xrightarrow[]{x\define {ch?a}} \conf'$ with the following sequence of transitions where the write leaves the buffer of $t_{ch}$ and is read by $t$:
	$\conf_{TSO} \xrightarrow{ t_{ch},u }. \xrightarrow {t,\rop{y_a}{r_x}}. \xrightarrow{ t,r_x\neq \rnew }. \xrightarrow{t_{ch}, u } .\xrightarrow[]{t,\rop{y_a}{\rtmp}}.  \xrightarrow{ t,\rtmp = \rnew }\conf'_{TSO}$.
	Note that the write buffer content of $t_{ch}$ is equivalent to the content of channel $ch$, it behaves the same throughout the run. 
	It follows that for any state $\conf'$ reachable by $\lcs$ from $\conf$ the corresponding state $\conf'_{TSO}$ is reachable by $\prog$ from $\conf_{TSO}$.
	
	The other direction is analogue, except that thread $t$ may not immediately perform the updates.
	This means the content of the channel may be split up between the buffers of $t$ and $t_{ch}$.
	However, any run of $\prog$ can be rearranged such that any write of $t$ immediately leaves the buffer and if it is read by $t_{ch}$, then this also happens immediately. Any write leaving the buffer of $t_{ch}$ can be delayed until immediately before the value is read by $t$. This results in a run where writes and reads and updates occur in the same sequences as the ones constructed above.
	The only exception is that it is possible that a write can leave the write buffer of a thread $t_{ch}$ without being read afterwards.
	In this case, there has been an earlier subsequence of transitions where $t_{ch}$, starting at $q_{ch}$, reads the corresponding writes and puts them into its own buffer before arriving back in $q_{ch}$. 
	We can simply remove this subsequence from the run and still arrive in the same configuration.
	It follows that any such changed run of $\prog$ has a corresponding run of $\lcs$ and thus reachability is implied in the other direction as well.
\section{Definition and Correctness of $\prog(\fsm_1\ldots \fsm_n)$}\label{sec:app:pspacehard}
We define thread $t=\langle \Q, \R, \Delta, q_{\init}\rangle $ as follows: It holds $\Q=\{q_\init, q \} \cup \{ q^a_i \mid i< n , a\in \Sigma \} \cup \{ q_\delta, q_i^\final \mid \delta\in \Delta_i,  i\leq n  \} $ and $\R = \{ r_q , r_i \mid q\in \Q_i , i\leq n  \}$. In addition, $\Q$ contains some helper states which we omit from the formal definition.

The set of transitions $\Delta$ contains an initialization sequence that starts at $q_\init$ and, using helper states, assigns each register $r_{q_i}$ a new value $r_{q_i} \newval$, then it checks that $r_{q_i}\neq r_{q'_i}$ holds for each pair  of states $q_i \neq q'_i$. Finally, it ensures that the current states are the initial states by assigning $r_i\define r_{q^i_\init}$. The sequence ends in state $q$.
Each transition $\delta = q_i\tstrans{a} q'_i  \in \Delta_i$ is simulated with transitions 
$q_{i-1}^a \tstrans{r_i=r_{q_i}} q_\delta, \; q_\delta\tstrans{r_i\define r_{q'_i}} q_{i}^a \in \Delta$ where $q^a_0= q^a_{n} =q$ and $i\leq n$.
Finally, we ensure that final states have been reached with $q \tstrans{r_1=r_{q_1}} q_{1}^\final, \; q_{i-1}^\final \tstrans{r_{i}=r_{q_i}} q_{i}^\final \in \Delta$ for $q_i\in \Q^F_i,  1<i\leq n$
It is easy to see that this construction is polynomial in the size of $\fsm_1\ldots\fsm_n$.
It remains to show that it is a correct reduction.

\medspace

Note that the program reaches $q_n^\final$ if and only if it can reach a configuration with state $q$ and $r_i=r_{q_i}$ as well as $q_i\in \Q^F_i$ for all $i\leq n$.
Correctness follows immediately from this fact and the following theorem:
\begin{theorem}
	$\prog ( \fsm_1 \ldots \fsm_n )$ has a run that contains $m+1$ configurations with state $q$ and that ends in a configuration where the state is $q$  and $r_i=r_{q_i}$ for $i\leq n$ holds if and only if 
	there is a word $w=a_1\ldots a_m$ such that there is a run $q^\init_i \tstrans{a_i} \ldots \tstrans{a_m} q_i$ for each $i\leq n$.
\end{theorem}
\begin{proof}
	We prove this with an induction over $m$.
	\paragraph*{Induction Basis ($m=0$):}
	The only  run of $\prog ( \fsm_1 \ldots \fsm_n )$ with only one configuration in $q$ is the one that consists of the initialization sequence. It ends in a configuration in state $q$ that satisfies $r_i=r_{q^\init_i}$ for $i\leq n$. 
	Only the runs $q^\init_1 \ldots q^\init_n$ without transitions correspond to the word $w=\epsilon$.
	\paragraph*{Induction Step $m\rightarrow m+1$:}
	Assume there is a word $w=a_i\ldots a_m a_{m+1}$ such that there exist runs $q^\init_i \tstrans{a_i} \ldots \tstrans{a_m} q'_i  \tstrans{a_{m+1}} q_i$ for each $i\leq n$.
	This is the case iff there exist runs $q^\init_i \tstrans{a_i} \ldots \tstrans{a_m} q'_i $ as well as transitions  $q'_i  \tstrans{a_{m+1}} q_i$ for each $i\leq n$.
	
	According to the induction hypothesis, this is the case if and only if the following holds:
	There are states $q'_1\ldots q'_n$ such that there is a run $\run$ of the program that contains $m+1$ configurations with state $q$ and ends in a configuration $\gamma$ with state $q$ that satisfies $r_i=r_{q'_i}$ for each $i\leq n$. 
	In addition there are transitions $q'_i  \tstrans{a_{m+1}} q_i \in \Delta_i$ for each $i\leq n$.
	According to the construction of $\Delta$, such transitions exist if and only if 
	there is a run $\run'$ where $q$ occurs only in the first and last configuration and the run goes from $\gamma$ to a configuration with state $q$ that satisfies $r_i=r_{q_i}$ for each $i\leq n$. 
	Note since $\run$ ends in $\gamma$ and $\run'$ starts in $\gamma$, we can append the runs. 
	A run $\run''=\run.\run'$ with $m+2$ configurations in state $q$ that ends in a configuration in state  that satisfies $r_i=r_{q^\init_i}$ for $i\leq n$ exists if and only if 
	there is a word $w=a_1\ldots a_{m+1}$ such that there is a run $q^\init_i \tstrans{a_i} \ldots \tstrans{a_{m+1}} q_i$ for each $i\leq n$.
\end{proof}

\section{Proof of \autoref{lem:scbab}}\label{sec:app_scbab}

Note that if a thread reads its own write that was issued in the same context, then it is irrelevant whether it reads from the buffer or the shared memory. Also, no other thread can read from that write during the context. 
We show that we can assume w.l.o.g. that updates happen at the end of the contexts.
Let $\run \in  \pmc(k)$  be a run of $\prog$.
We show that we can move all the update transitions to the end of a context without destroying the correct $\prog$ semantics.
We examine any update followed by an operation $\op$ that is not an atomic read write of the same thread in $\run$:
$\conf \trans{t}{u} \conf'' \trans{t}{\op} \conf'$.
It follows immediately from the  \tso\ semantics that there is a $\conf_1$ such that the following is correct as well:
$\conf \trans{t}{\op} \conf_1 \trans{t}{u} \conf'$.
This means that any configuration reachable by a run $\run \in \pmc(k)$ can also be reached by a run $\run' \in \pmc(k)$ that is equivalent to $\run$ except the updates are moved to the end of the context. 

The exception to this are atomic read write operations. 
We cannot move updates past them since they require an empty buffer.
We adjust to that by not simulating the exact semantics of atomic read write. 
$\scbab$ does not require an empty buffer, merely the assurance that the current buffer will be emptied in the current context.
We introduce two kinds of modified atomic read write semantics, both are only enabled if the current content of the write buffer will be emptied in the current context.
A \emph{buffer arw} $\arw{x}{r_1}{r_2}$ reads the latest write on $x$ that is currently in the write buffer and, if its value is equal to $r_1$, changes it to the value of $r_2$.
A \emph{memory arw } does the same directly on the shared memory if there is no write on $x$ in the buffer.
Note that if the buffer is empty, only the memory arw can be executed and in that case, it has the same semantics as a standard arw.
This means any \tso\ run is still possible with the modified arw semantics.
Further, we can show that these new arw semantics are reachability equivalent to the standard arw semantics:
Given a run $\run$ with the modified arw semantics, 
If there is an arw, then that means the buffer will be empty in some later point in the context.
This means we can take the updates between the arw and that point that have corresponding writes before the arw and move them backwards past the arw. 
Doing this may cause a buffer arw to become a memory arw.
Afterwards, the buffer is empty when the arw is executed. This means we can replace it with a standard arw.
We can replace all modified arw operations in $\run$ with a standard arw in that way.
It follows that for every run with modified arw operations, there is a corresponding run with standard arw operations that has the same number of $\pmc$ contexts and reaches the same state.
We omit the formal details of this.
This means the modified arw semantics employed by $\scbab$ are reachability equivalent with the standard arw semantics.
\medskip

For state reachability, we can restrict ourselves to examining such runs where writes only leave the buffer at the end of each context. 
In any such \pmc\ run, a read of $x$ in thread $t$
\begin{itemize}
	\item reads the last write in $t$ on $x$ if there is such a write in its write buffer.
	\item reads the last write on on $x$ that left any buffer in an earlier context if the buffer of $t$ contains no write on $x$.
\end{itemize}

It follows that it is sufficient to track the last writes on each variable that leave the buffers in each context as well as the latest writes on each variable in the buffer of the active thread since no other write can be read. 
Those are exactly the writes stored in the constructed transition system. 
Instead of storing the whole buffer, it abstracts from it by only storing the last entries on every variable. We know whether the write buffer contains such an entry since we guess for each write operation at the moment it is issued in which context  it will update the memory. We store this guess for the last writes of each variable for every thread.
It is clear that the transitions are consistent with the \tso\ semantics and that any \tso\ run where updates occur at the end of the contexts has an equivalent run of $\scbab$.
It follows from this observation that the construction correctly models state reachability of processor-memory-bounded \tso\ behaviour.
\section{Proof of \autoref{thm:relAB}}\label{sec:app:relAB}
First, we show that a state reachable by $\scbab$ is also reachable by $\relAB$. 
This is the case because $\relAB$ abstracts from the domain:
Let $\run$ be a run of $\scbab$. It follows from \autoref{lem:abstract}, that we can construct a run $\run'$ of $\relAB$ by replacing every step $(q,\funX)\tstrans{op}(q',\funX')$ of $\run$ with  $(q,\Rel_{\funX})\tstrans{op}(q',Rel_{\funX'})$.
It follows that if $\scbab$ reaches a configuration $(q,\funX)$, there is a run of $\relAB$ that reaches $(q,\Rel_{\funX})$.

It remains to show that for any  state $(q_a,\Rel_a)$ reachable by $\relAB$ , there is a configuration $(q_a,\funX_a)$  reachable by $\scbab$
such that $\Rel_a$ is consistent with $\funX_a$.
Here, the main challenge is that while $\relAB$ models the ordering of the variables, it does not keep track of the precise distances between variables. So how can we decide whether a relation $<_n$ holds?
If there is a non-zero value $d$ in a reachable configuration of $\scbab$, then $d$ and all larger values were generated by $x\newval$ operations earlier in the run. We can change the run so that these operations could also assign even larger values. This increases the distances between variables.
For state reachability we can assume that any distance between variables with different values is large enough. 
To this effect, we use the following lemma, which shows that if a configuration is reachable, then any configuration with larger distances is also reachable. 
We can always increase all values in a run greater or equal than some $d$ by the same value $c$.
\begin{lemma}\label{lem:increasingdistances}
	For any $0<c,d\in \D$ holds that if a configuration $(q_a,\funX)$ is reachable by $\scbab$, then $(q_a,\funX')$ with $\funX'(x)=\funX(x)$ if $\funX(x)<d$ and $\funX'(x)=\funX(x)+c$ if $\funX(x)\geq d$ is also reachable.
	It holds $\Rel_\funX=\Rel_{\funX'}$.
\end{lemma}
\begin{proof}
	
	Let $\run$ be a run of $\scbab$ ending in $(q,\funX)$.
	The run contains a set of distinct values which we list in ascending order: $d_1< \ldots < d_m$. 
	Let run $\run'$ be the same run as $\run$ except  all occurrences of any value $d'$ with $d'\geq d$ are replaced by $d'+c$:
	Let $d_i$ be the first value with $d_i\geq d$.
	The run $\run'$ contains the set of values: $d_1< \ldots < d_{i-1}< d_i+c < \ldots < d_m+c$. 
	We show that $\run'$ is a correct run with the desired property with an induction over the run $\run$:
	\noindent\textbf{Induction basis.} 
	The initial configuration of $\run$ is the same as in $\run'$: $(q_\init, \funX_{\init})$. This is the case since every value of $\funX_{\init}$ is 0 and we only change values starting at $d>0$.
	
	\noindent\textbf{Induction hypothesis.} If $\run$ ends in $(q,\funX)$, then $\run'$ ends in $(q,\funX')$ with $\Rel_{\funX}=\Rel_{\funX'}$.
	
	\noindent\textbf{Induction step.} There is a run $\run$ ending in $(q,\funX)$. 
	We add $(q,\funX)\tstrans{x<_n y }(q_a,\funX_a)$ to $\run$. 
	According to the induction hypothesis, there is a run $\run'$ ending in $(q,\funX')$.
	Let $\funX'_a$ be $\funX$ with the values replaced according to the lemma.
	It follows from the semantics of the operation that for any $z\in \X$ holds $\funX(z)=\funX_a(z)$. This means there is some $d_j$ such that $d_j=\funX(z)=\funX_a(z)$. If $j<i$, then it holds $d_j=\funX'(z)=\funX'_a(z)$. If $j\geq i$, then $d_j+c=\funX'(z)=\funX'_a(z)$. It follows $\funX'=\funX'_a$.
	From $\funX(x)<_n \funX(y)$ follows that there are $j,k$ such that 
	$\funX(x)=d_j  <_n \funX(y)=d_{j+k}$. It holds either $\funX'(x)=d_j$,  $\funX'(y)=d_{j+k}$, or $\funX'(x)=d_j$, $\funX'(y)=d_{j+k}+c$, or $\funX'(x)=d_j+c$, $\funX'(y)=d_{j+k}+c$. In each case $\funX'(x)<_n \funX'(y)$ still holds.
	It follows that we can add $(q,\funX')\tstrans{x<_n y }(q_a,\funX'_a)$ to $\run'$.
	For any other operation, the proof is analogue.		
	
	When constructing $\funX'_a$ from $\funX_a$, we see that if two value are the same, they stay the same. If two values are different, then $c$ is added to either none or both or only to the larger value. This means $\funX'_a$ satisfies the same relations of $\Relinit$ as $\funX_a$.
	It follows that $\Rel_{\funX_a}=\Rel_{\funX'_a}$.
	
\end{proof}
We apply an induction over the length of the run that reaches $(q_a,\Rel_a)$.

\noindent\textbf{Induction basis} The run of $\relAB$ consists only of the initial state $(q_0, \Rel_{\funXinit})$. The initial configuration of $\scbab$ is $(q_0, \funXinit)$ and ${\funXinit}$ is consistent with $\Rel_{\funXinit}$.

\noindent\textbf{Induction hypothesis} If a run of  $\relAB$ ends in $(q_a,\Rel_a)$, there is a run of $\scbab$ ending in $(q_a,\funX_a)$ such  that $\funX_a$ is consistent with $\Rel_a$.

\noindent\textbf{Induction step} We add a step $(q_a,\Rel_a)\tstrans{op}(q_b,\Rel_b)$ to the run. 
If the operation is $x\define y$ then $(q_a,\funX_a)\tstrans{x\define y}(q_b,\funX_a[x\leftarrow\funX_a(y)])$.
From the induction hypothesis and the semantics of $\relAB$ follows that $\funX_a[x\leftarrow\funX_a(y)]$ is consistent with $\Rel_b$.

If the operation is a relation in $\Relinit$, then transition $(q_a,\funX_a)\tstrans{\rl{x}{y}}(q_b,\funX_a)$ is possible since $\relname{rel}_a ({x},{y})\in \Rel_a$ and $\funX_a$ is consistent with $\Rel_a$. Since relations remain unchanged by the transition it holds $\Rel_a=\Rel_b$ and thus $\funX_a$ satisfies $\Rel_b$.

Assume the operation is a relation $x<_{n} y$ not satisfied by $\funX_a$. 
We can apply \autoref{lem:increasingdistances} to $(q_a,\funX_a)$ with $d=\funX_a(y)$.
This is possible since according to the semantics of $\relAB$, it holds $x<^a y$ and thus $\funX_a(y)>0$. We add a $c$ such that $\funX_a(y)+c=\funX_a(x)+n+1$. 
The resulting assignment $\funX_a'$ satisfies $x<_n y$ since $\funX_a'(y)=\funX_a(y)+c>\funX_a'(x)+n$.
It follows that $(q_a,\funX_a')\tstrans{x<_n y}(q_b,\funX_a')$ is possible.
Since relations remain unchanged by the transition, it holds $\Rel_a=\Rel_b$ and thus $\funX_a'$ is consistent with $\Rel_b$. Relation $x\leq_{n} y$ is analogue.

Assume the operation is an assignment $x\newval$. 
If it assigns a value equal to some variable $y$ indicated by $x=_b y$, then the behaviour is the same as $x\define y$ for which the property has already been proven.
Assume $x\newval$ assigns a new value with $x\neq^b y$ for all $y\in \X$ with $x\neq y$. 
Let $y,z\in \X$ be the closest variables to $x$ such that $y<^b x<^b z$.
If $ \funX_a (y)+1<\funX_a (z) $, then it holds $(q_a,\funX_a)\tstrans{x\newval} (q_b,\funX_a[ x\leftarrow \funX_a(y)+1 ]) $.
For any pair of variables different from $x$, their relations remain unchanged by the operation.
The relations of $y$ and $z$ to $x$ in $\funX_a[ x\leftarrow \funX_a(y)+1]$ are the same as in $\Rel_b$.
For any other variable, its relations to $x$ is determined by its relations to $y$ or $z$ (depending on whether they are smaller or larger). 
It follows that $\funX_a[ x\leftarrow \funX_a(y)+1]$ is consistent with $\Rel_b$ and thus the hypothesis remains satisfied.
If $\funX_a (y)+1= \funX_a (z)$, then we apply \autoref{lem:increasingdistances} to increase $d=\funX_a(z)$ by $c=1$. For the resulting assignment $\funX_a'$, it holds $ \funX_a '(y)+1<\funX_a' (z) $ and thus $(q_a,\funX_a')\tstrans{x\newval} (q_b,\funX_a'[ x\leftarrow \funX_a'(y)+1 ]) $ and $ \funX_a'[ x\leftarrow \funX_a'(y)+1 ]$ is consistent with $\Rel_b$.

If there is no larger variable $z$ with $x<_b z$ then it is easy to see that $(q_a,\funX_a)\tstrans{x\newval}(q_b,\funX_a[ x\leftarrow \funX_a(y)+1])$ and $\funX_a[ x\leftarrow \funX_a(y)+1]$ is consistent with $\Rel_b$.
We can ensure that $x$ is never smaller than all other variables by adding a new variable $x_0$  to $\scbab$ that is never used and thus remains 0. We add a restriction to the transition for operation $x\newval$ that ensures $x<x_0$ does not hold.
\end{document}